\documentclass{LMCS}

\def\doi{9(2:13)2013}
\lmcsheading%
{\doi}
{1--17}
{}
{}
{Oct.~22, 2012}
{Jun.~28, 2013}
{}

\usepackage{amsthm}
\usepackage{amsfonts}
\usepackage{amsmath}
\usepackage{amssymb}

\usepackage{hyperref,enumerate}

\usepackage{theory}
\usepackage{theorems}

\usepackage{color}

\begin{document}

\title[Lower Bound on Weights of Large Degree Threshold Functions]{Lower Bound on Weights of Large Degree Threshold Functions\rsuper*}


\date{}

\author[V.~V.~Podolskii]{Vladimir V. Podolskii}
\address{Steklov Mathematical Institute, Moscow}
\email{\href{mailto:podolskii@mi.ras.ru}{podolskii@mi.ras.ru} }

\thanks{The work is supported by the Russian Foundation for Basic Research
and the programme ``Leading Scientific Schools'' (grant no. NSh-5593.2012.1).
}


\ACMCCS{[{\bf Theory of computation}]: Models of computation;
  Computational complexity and cryptography---Circuit complexity}

\titlecomment{{\lsuper*}A preliminary version of this paper appeared
  in the proceedings of CiE 2012 conference, Springer LNCS Proceedings
  volume 7318}



\begin{abstract}
An integer polynomial $p$ of $n$ variables
is called a \emph{threshold gate} for a Boolean function $f$
of $n$ variables if for all $x \in \zoon$ $f(x)=1$ if and only if $p(x)\geq 0$.
The \emph{weight} of a threshold gate is the sum of its absolute values.

In this paper we study how large a weight might be needed if we fix some function and some threshold degree.
We prove $2^{\Omega(2^{2n/5})}$ lower bound on this value.
The best previous bound was $2^{\Omega(2^{n/8})}$ (Podolskii, 2009).

In addition we present substantially simpler proof of the weaker $2^{\Omega(2^{n/4})}$ lower bound.
This proof is conceptually similar to other proofs of the bounds on weights of nonlinear threshold gates,
but avoids a lot of technical details arising in other proofs.
We hope that this proof will help to show the ideas behind the construction used to prove these lower bounds.
\end{abstract}

\maketitle

\section{Introduction} \label{sec.introduction}

Let $f \colon \zoon \to \zoo$ be a Boolean function.
\emph{A threshold gate} for the Boolean function $f$ is an integer polynomial $p(x)$ of $n$ variables $x=(x_1, \ldots, x_n)$
such that for any $x \in \zoon$ we have $f(x)= 1$ if and only if $p(x) \geq 0$.
In other words, for all $x \in \zoon$ it is true that $f(x) = \sgn p(x)$,
where we adopt the following definition of the sign function: $\sgn(t)=1$ if $t\geq 0$ and $\sgn(t)=0$ otherwise.

Thus, threshold gates are just representations of Boolean functions as the signs of the polynomials.
The formal study of such representations started in $1968$ with the seminal monograph of Minsky and Papert~\cite{minsky88perceptrons}.
Since then representations of this form found a lot of applications in circuit complexity, structural complexity,
learning theory and communication complexity (see, for \mbox{example~\cite{razborov92thr-circuits,KS01dnf,beigel94perceptrons,dual-survey}}).

Two key complexity measures of threshold gates are their degree and their weight.
\emph{The degree} $\deg p$ of a threshold gate $p$
is just the degree of the polynomial. \emph{The weight} $W(p)$ of a threshold gate $p$ is the sum of absolute values
of all its coefficients.

The complexity measures of a Boolean function $f$ related to these complexity measures of threshold gates
are the minimal threshold degree of a threshold gate for $f$ which we denote by $\signdeg f$ and call \emph{the threshold degree} and
the minimal weight of a threshold gate for $f$.
Both of these complexity measures play an important role in theoretical computer science (see the references above).
In this paper we are interested in the minimal possible value of the weight of a threshold gate for some function $f$
when the degree of the threshold gate is bounded.
It is convenient to denote by $W(f,d)$ the minimal weight of a threshold gate of degree at most $d$ for $f$.
Note that this value is defined only if $d \geq \signdeg f$.
It is also not hard to see that for all $f$ we have $\signdeg f \leq n$ and $W(f,n) \leq 2^{O(n)}$
(just consider the polynomial $p$ such that $p(x)=f(x)$ for all $x$).

The first results on the value of $W(f,d)$ were proven for $d=1$.
In~\cite{muroga61} (see also~\cite{muroga71} and~\cite{hastad94weights}) it was proven that
for all $f$ with $\signdeg f = 1$ it is true that $W(f,1) = n^{O(n)}$.
For a long time only lower bounds of the form $W(f,1) = 2^{\Omega(n)}$ were known
(see~\cite{myhill-kautz61} for one of the early results).
Tight lower bound was proven in~\cite{hastad94weights}, that is the function $f$ with $\signdeg f = 1$
was constructed such that $W(f,1) = n^{\Omega(n)}$.

Concerning higher degree $d$, upper bound can be easily extended from the case $d=1$.
Namely for all $f$ with $\signdeg f \leq d$ it is true (and easy to see) that $W(f,d) = n^{O(dn^{d})}$ (see~\cite{saks93slicing,buhrman07pp-upp,ppi09weights}).
Note that this upper bound is much worse than for the case $d=1$.
Concerning the lower bounds, it is rather easy to see that the bound $n^{\Omega(n)}$ can be translated from
the case $d=1$ to arbitrary $d$ (see the discussion preceding Theorem~\ref{theorem.high.deg}).
The first lower bound improving this was given in~\cite{ppi09weights} and showed that the upper bound stated above is actually tight for constant $d$.
That is, for any constant $d$ the function $f$ of the threshold degree $d$ was constructed in~\cite{ppi09weights} such that $W(f,d) = n^{\Omega(n^{d})}$
(constant in $\Omega$ here depends on $d$).
It is implicit in~\cite{ppi09weights} though that the argument works for
nonconstant $d$ also and the resulting lower bound (with the dependence on $d$)
is
\begin{equation} \label{eq.old}
\left(\frac{n}{d}\right)^{\frac 12 (\frac{n}{2d})^{d} - o((\frac{n}{d})^d)}.
\end{equation}
For this result another proof was given in~\cite{bhps10}.
Some other results on large degree threshold gates which are not directly connected to the problem
we consider have appeared in~\cite{GHR92,beigel94perceptrons,ServedioTT12COLT,trudy11weights,mn10weights,DiakonikolasSTW10CCC}.

Thus it turns out that the required weight grows with the growth of the degree $d$.
In this paper we are interested in how large it might grow (note that for $d=n$ the weight is small again: $W(f,n) = 2^{O(n)}$).
That is we study the value
$$
W = \max_{d} \max_{f \colon \signdeg(f)\leq d} W(f,d).
$$
The lower bound~\eqref{eq.old} works even for $d$ depending linearly on $n$ and so gives doubly exponential lower bound on this value.
But it works only for $d \leq (n-c)/32$,
where $c$ is some constant, so the best lower bound we get from~\cite{ppi09weights} is $2^{\Omega(2^{n/8})}$.

In this paper we prove the following bound.
\begin{thm}
$W \geq 2^{\Omega(2^{2n/5})}$.
\end{thm}
We note that the best upper bound known is simple $2^{O(n2^n)}$ (this can be deduced from the upper bound for the case
$d=1$ and the fact that there are at most $2^n$ monomials).

To prove our lower bound we adopt the strategy of~\cite{ppi09weights}
and provide a unified treatment of the argument of that paper.
In short, the proof strategy is as follows.
Starting from some function of the threshold degree $1$ (with some additional properties) that requires large weight when represented by degree-$1$
threshold gates we construct its ``$d$-dimensional'' generalization in a very specific way.
For this generalization we are able to prove a lower bound for degree-$d$ threshold gates
and due to the specific features of our generalization we can prove a strong lower bound.

In the paper~\cite{ppi09weights} the construction of the function starts with the function constructed by H{\aa}stad in~\cite{hastad94weights}
to prove the optimal lower bound for the case $d=1$. This helps to get $n$ in the base of the exponent in the lower bound
and thus to prove a strong lower bound for the case of constant $d$.
On the other hand, H{\aa}stad's function is very complicated and has
desired properties only for large enough number of variables ($16$ variables).
This does not allow us to prove a lower bound for $d$ close to $n$.
In this paper we start with a much simpler functions having required properties starting from just $3$ variables.
With this function we cannot get $n$ in the base of the exponent, but on the other hand we are able now to
prove bounds for much larger $d$ and thus to get better lower bound on $W$.

We start exposition of our result by giving a simpler proof of the weaker bound of $2^{\Omega(2^{n/4})}$.
In this proof we are able to avoid a lot of technical complications
arising in the proof of~\cite{ppi09weights} and make the function for which we prove the bound much simpler (here we use as a starting function of the threshold degree $1$ well known ``greater than'' function).
We hope that this makes the proof easier to read
and helps to show the ideas behind the construction which were not very clear in~\cite{ppi09weights}.

After that we define another starting function and explain how to change the proof
to get $W \geq 2^{\Omega(2^{2n/5})}$ lower bound.
The idea here is not only that we can prove the bound for larger $d$,
but also that, roughly speaking, choosing the good function we can remove the constant $2$ from the denominator of the term $\left(\frac{n}{2d}\right)$
in the exponent in the bound~\eqref{eq.old}.

Besides representation of Boolean functions as $f \colon \zoon \to \zoo$, also representation
of the form $f \colon \moon \to \moo$ turns out to be useful in complexity theory.
Here $-1$ corresponds to ``true'' and $1$ corresponds to ``false''.
For this representation we can also consider threshold gates and also define corresponding measures of the functions.
Note that we can switch from one representation to another one by a simple linear transform.
Thus the threshold degree of the function does not depend on the representation
and the threshold weight may change only by the $2^n$ multiplicative factor
(see~\cite{KP98threshold} for more information on the relations between the threshold weights in these two settings).
Since this factor is very small compared to our lower bound,
our result is true for both representations and in the proof we can choose the one of two presentations
of Boolean functions which is more convenient to us.
For the proof of the weaker bound we will use $\{0,1\}$ variables and for the stronger one --- $\{-1,1\}$ variables.

We note in the conclusion that if we have the lower bound $S$ on the minimal weight for the function of $n$ variables
and for the degree $d$, it is easy to translate it to exactly the same bound $S$ for $n'=n+c$ and $d'=d+c$ for any $c$,
even depending on $n$ (see, for example,~\cite{mn10weights}, Corollary~$1$).
This observation allows us to deduce strong lower bounds on weights of threshold gates
of degree close to $n$.
\begin{thm} \label{theorem.high.deg}
For any $\eps >0$ and $d \leq (1- \eps )n$ there is an explicit function $f$ such that $W(f,d) = 2^{2^{\Omega(n)}}$.
For any $d \leq n - 2(1+ \eps) \log n$ there is an explicit function $f$ such that $W(f,d) = 2^{\Omega(n^{1+\eps})}$.
\end{thm}

The rest of the paper is devoted to the formulation and the proof of our results.
In Sections~\ref{sec.function} and~\ref{sec.result} we give a simple proof for the weaker bound:
in the former we construct the function for which in the latter we prove the lower bound.
In Section~\ref{sec.better} we explain how to change the proof~to~give~the~stronger~bound.

\section{Construction of the Function} \label{sec.function}

In this section we present the construction of the function for which we prove a weaker form of our bound.
Our function is the generalization of the $\GT$ function.

\begin{definition}
For Boolean $x, y \in \zoo^k$ let $\GT(x,y) = 1$
iff $x \geq y$, where $x = (x_{1}, \ldots, x_{k})$ and $y = (y_{1}, \ldots, y_{k})$ are considered as binary representations of integers
with $x_{k}$ and $y_{k}$ being the most significant bits.
\end{definition}

Our function will depend on $n=2m$ variables
$$
(x,y) = (x_1, \ldots, x_{m}, y_{1}, \ldots, y_{m}) \in \zoon.
$$
Let us fix some $k_1, \ldots, k_{d}$ such that $\sum_{i=1}^{d} k_{i}  = m$ and partition the input variables
$x$ and $y$ in $d$ groups of size $k_1, \ldots, k_d$, that is
$$
(x,y)=(x^{1}, x^{2}, \ldots, x^{d}, y^{1}, y^{2}, \ldots,  y^{d}),
$$
where $x^{i},y^{i} \in \zoo^{k_i}$ for all $i$.

Let us denote by $[k]$ the set $\{1,\ldots,k\}$.
Let us denote by $<_{1}$ the following ordering of the set $[k]$:
$
1,2,3, \ldots, k-1, k,
$
and by $<_{0}$ the reverse ordering:
$
k, k-1, k-2, \ldots, 2, 1.
$
We will use these orders on the sets $[k_1], \ldots, [k_d]$.
It will always be clear from the context which set we consider.

Let us denote by $\num_{i} l$ the ordinal number of $l \in [k]$ w.r.t. the order $<_{i}$.

To define our function we need to define a specific order on the set $K = [k_1] \times \ldots \times [k_d]$.
The construction below is essentially the same as in~\cite{ppi09weights}.
Our order will be similar to the lexicographic one, that is
to compare two tuples from $K$ we will compare their components one by one
until we find the difference. But as opposed to the lexicographic
order, where each component of the tuples is compared w.r.t. the same ordering,
in our order of the tuples components might be compared w.r.t. different orderings.
Moreover, the ordering in which we compare the current component depends not only
on the ordinal number of the component but on the values of previous components of the tuples.

Formally, suppose we want to compare tuples $\alpha = (\alpha_1, \ldots, \alpha_d) \in K$ and
$\beta = (\beta_1, \ldots, \beta_d) \in K$.
First we compare $\alpha_1$ and $\beta_1$ w.r.t. the ordering $<_1$.
If they are not equal then we have already compared the tuples:
the larger the first component is the larger the tuple is.
If they are equal we proceed to the second components.
To compare them we use the following recursive rule to choose the next order.

Assume that the order $<_{i_{l}}$ to compare the $l$th components of the tuples
is already determined and it happens that $\alpha_l=\beta_l$.
The order to compare $(l+1)$st components is determined by the ordinal
number of $\alpha_{l}$ (which coincides with $\beta_l$ by the assumption)
w.r.t. the order  $<_{i_l}$.
Namely,
\begin{equation} \label{eq.order}
i_{l+1} = \num_{i_l} \alpha_l\ (\modd 2).
\end{equation}
In other words, we compare the $(l+1)$st coordinates w.r.t. the order $<_0$
if $\alpha_l$ has even ordinal number w.r.t. the order $<_{i_l}$
and we compare $(l+1)$st coordinates w.r.t. the order $<_1$ otherwise.

To say it the other way, we associate with the coordinates of any tuple $\alpha = (\alpha_1, \ldots, \alpha_d) \in K$
the orders $<_{i_1}, \ldots, <_{i_d}$ according to the rule~\eqref{eq.order}.
We use these orders to compare the coordinates of $\alpha$ with the coordinates of other tuples.
Note that for two tuples $\alpha = (\alpha_1, \ldots, \alpha_d)$ and
$\beta = (\beta_1, \ldots, \beta_d)$ the orders corresponding to their components
coincides until we meet the first difference.
After the first difference the orders corresponding to the components might
be different in $\alpha$ and $\beta$ but we do not need to compare the coordinates any further.

Let us denote by
$
\num_{\alpha,l} \alpha_l
$
the ordinal number of the $l$th component of $\alpha$ w.r.t.
the corresponding order.

It is not hard to describe the order we constructed in the case $d=2$.
Let us represent the pairs $(\alpha_1, \alpha_2)$ by the points on the plane (see Figure).

{
\setlength{\unitlength}{0.6pt}
\begin{figure}
\begin{center}
\begin{picture}(330,330)(-30,-30)
\thicklines
\put(0,0){\vector(0,1){300}}
\put(0,0){\vector(1,0){300}}

\put(30,30){\vector(0,1){240}}
\put(60,270){\vector(0,-1){240}}
\put(90,30){\vector(0,1){240}}

\put(270,270){\vector(0,-1){240}}
\put(240,30){\vector(0,1){240}}

\put(30,270){\vector(1,0){30}}
\put(60,30){\vector(1,0){30}}
\put(90,270){\vector(1,0){15}}
\put(240,270){\vector(1,0){30}}
\put(225,30){\vector(1,0){15}}

\put(30,-2){\line(0,1){4}}
\put(60,-2){\line(0,1){4}}
\put(90,-2){\line(0,1){4}}
\put(240,-2){\line(0,1){4}}
\put(270,-2){\line(0,1){4}}

\put(-2,30){\line(1,0){4}}
\put(-2,60){\line(1,0){4}}
\put(-2,270){\line(1,0){4}}


\put(295,-16){$\alpha_1$}
\put(-25,295){$\alpha_2$}

\put(28,-20){$1$}
\put(58,-20){$2$}
\put(88,-20){$3$}
\put(268,-20){$k_1$}
\put(160,-20){$\ldots$}

\put(160,30){$\cdots$}

\put(-24,27){$1$}
\put(-24,57){$2$}
\put(-24,267){$k_2$}
\end{picture}
\end{center}
\begin{center} Case $d=2$
\end{center}
\end{figure}
}

The first component $\alpha_1$
is associated with the horizontal axis, and the second component
$\alpha_2$ -- with the vertical axis.
The arrows indicate the direction
from the smaller pairs to the larger ones.

In the case $d>2$ the constructed order is not so easy to describe.
However, as our proof goes by induction, we always
consider only two consecutive coordinates of the tuple. Thus we will be in the situation
that is very similar to the case
$d=2$. The only difference is that the order on the first
coordinate might be different from $<_{1}$.

Now we can define our function.

\begin{definition} \label{def.func}
For a given $(x,y) = (x^1, x^2, \ldots, x^d, y^1, y^2, \ldots, y^d)$, where
$x^i = (x^i_1, \ldots, x^i_{k_i}), y^i = (y^i_1, \ldots, y^i_{k_i}) \in \zoo^{k_i}$
let $\alpha = (\alpha_1, \ldots, \alpha_d) \in K$ be the largest tuple
w.r.t. the introduced order such that
$
\prod_{i=1}^{d} (x^{i}_{\alpha_i} - y^{i}_{\alpha_i}) \neq 0.
$
Then let
$$
f(x^{1}, \ldots, x^{d}, y^{1}, \ldots,  y^{d}) = \sgn \prod_{i=1}^{d} (x^{i}_{\alpha_i} - y^{i}_{\alpha_i}).
$$
If there is no such $\alpha$ let
$
f(x^{1}, \ldots, x^{d}, y^{1}, \ldots,  y^{d}) = 1.
$
\end{definition}

Note that if $d=1$ our function is exactly the $\GT$ function.

\section{$2^{\Omega(2^{n/4})}$ Lower Bound} \label{sec.result}

First we note that our function is computable by a degree $d$ threshold gate.
\begin{lemma} \label{lemma.degree}
$\signdeg(f) \leq d$.
\end{lemma}

\begin{proof}
Let
$$
\alpha^1, \alpha^2, \ldots, \alpha^{|K|}
$$
be the list of all elements of $K$ written w.r.t.
the introduced order.
For any $\alpha^{j} = (\alpha^{j}_{1}, \ldots, \alpha^{j}_{d})$ let $t_j$ be the product
$$
t_j = \prod_{i=1}^{d} (x^{i}_{\alpha^j_i} - y^{i}_{\alpha^j_i}).
$$
Note that for any $j$ and any input $(x,y)$ we have $|t_{j}| \leq 1$.

Consider the polynomial
$$
p(x,y) = \sum_{j=1}^{|K|} 2^{j} t_{j}.
$$
We claim that it sign-represents our function.

Indeed, let us fix some input $(x,y)$. If for this input all $t_{j}$ are zeros then
$f(x,y) = 1$ and $p(x,y) = 0$.
Otherwise, let $\alpha^{j}$ be maximal w.r.t. the introduced order such that $t_{j} \neq 0$.
Then all $t_{i}$ for $i > j$ are zeros, and the coefficients of $p$ are chosen in such way
that the contribution of $t_{j}$ is greater than the sum of contributions of all $t_{i}$ for $i < j$.
Thus the sign of $p(x,y)$ is the same as the sign of $t_{j}$.
\end{proof}

Now we proceed to the main result of this section.

\begin{thm} \label{theorem.main}
Let $k_i \geq 2$ be even for $i < d$ and $k_d \geq 3$. Then
$$
\signdeg(f) = d
$$
$$
W(f,d)\geq 2^{(k_d-2)\prod\limits_{i=1}^{d-1}k_i\ - d}
$$
\end{thm}

\begin{remark}
We can state an analogous theorem for arbitrary $k_i\geq 2$ for $i<d$
and not only for even. However, with this assumption the proof and the bound are cleaner and at the same time
the theorem still gives $2^{\Omega(2^{n/4})}$ bound.
\end{remark}

First we give some corollaries of Theorem~\ref{theorem.main} and then proceed to the proof.

If we let $k_i = 2$ for $i < d$ and $k_d = 3$ in Theorem~\ref{theorem.main} we get
\begin{equation*} 
W \geq W(f,\frac{n-2}{4}) \geq 2^{2^{\frac{n-6}{4}} - d}.
\end{equation*}
More generally,
\begin{corollary}
For all $n$ and all $d \leq \frac{n-2}{4}$ there exists (explicit) function $f$ of $n$ variables such that
$$
W(f,d) \geq 2^{(2\lfloor \frac{n-2}{4d} \rfloor)^{d} - (2\lfloor \frac{n-2}{4d} \rfloor)^{d-1} - d}.
$$
\end{corollary}
\begin{proof}
Let $k_i = 2 \lfloor \frac{n-2}{4d} \rfloor$ for $i<d$, $k_d = 2 \lfloor \frac{n-2}{4d} \rfloor + 1$, consider
the function from Theorem~\ref{theorem.main} on the first $2 \sum_{i=1}^d k_i$
variables and leave other variables inessential.
\end{proof}

Analogously we can prove other bounds. For example we can prove the following bound with a simpler formulation
$$
W(f,d) \geq 2^{(2\lfloor \frac{n-4}{4d} \rfloor)^{d} - d}.
$$

\subsection{Proof of Theorem~\ref{theorem.main}}

Let us consider an arbitrary threshold gate $p$
for $f$ of the degree at most $d$.
That is, for any $x, y \in \zoom$ we have
$$
f(x,y) = \sgn (p(x,y)).
$$
It will be convenient for us to work in variables
\begin{equation} \label{eq.new.war}
u^{i}_{j} = x^{i}_j - y^{i}_{j},\ \ \ \ \ \ \ v^{i}_{j} = x^{i}_j + y^{i}_{j}.
\end{equation}
So after substituting
$$
x^{i}_j = ({u^{i}_j + v^{i}_j})/{2},\ \ \ \ y^{i}_j = ({u^{i}_j - v^{i}_j})/{2}
$$
and multiplying the polynomial by $2^d$ to make the coefficients integer we obtain
the polynomial $p'$ in the variables $u^{i}_j,\ v^{i}_j$ that sign-represents $f$. That is
$$
f(x,y) = \sgn(p'(x-y,x+y)).
$$
It is easy to see that the weight of the new polynomial is almost the same as the weight of $p$
(compared to the value of our bound). Namely, we have the following bound.

\begin{lemma} \label{lemma.basis_change}
$W(p') \leq 2^d W(p)$.
\end{lemma}

\begin{proof}
Consider one monomial $g$ of $p$ of the degree $l \leq d$.
We can think of the transformation above as of substituting each variable by the sum (or the difference)
of two variables and then multiplying everything by $2^{d-l}$.
After opening the brackets we will have $2^{l}$ new monomials each of the weight $2^{d-l}$,
so the overall weight of the new monomials coming from the monomial $g$ is $2^{d}$ times the coefficient of $g$.
Since this happens for all monomials, the upper bound on the weight of the polynomial $p'$ as stated follows.
\end{proof}

\begin{remark}
A similar bound holds in the other direction too, but we do not need it.
\end{remark}

Now we have to prove that
$$
W(p') \geq 2^{(k_d-2)\prod_{i=1}^{d-1}k_i}.
$$
First we will prove that we can assume that $p'$ has a nice structure.
Lemmas similar to the next one appeared in~\cite{hastad94weights,ppi09weights} (see \cite{mn10weights} for a more general version).
\begin{lemma} \label{lemma.symmetry}
If we substitute by $0$ all coefficients of the monomials of $p'$ in which the variables from one of the groups $u^{1}, \ldots, u^{d}$ do not appear,
the resulting polynomial $q$ will also sign-represent $f$.
\end{lemma}

\begin{proof}
Let $I$ be the set of inputs $(x,y)$ to $f$ such that $x^i \neq y^i$ for all $i = 1, \ldots, d$.
This condition implies that for each such input there is an $\alpha = (\alpha_1, \ldots, \alpha_d) \in K$ such
that
$$
\prod_{i=1}^{d} (x^{i}_{\alpha_i} - y^{i}_{\alpha_i}) \neq 0.
$$
This means that the value of the function is determined by the sign of this product for the largest such $\alpha$.

Consider an arbitrary input $(x,y)$ from $I$ and
consider an arbitrary $i \in [d]$. Let us permute the variables $x^{i}$ with the variables $y^{i}$.
It is not hard to see that after such permutation the function $f$ changes sign, and so $p'$ also should change sign.

Now let us see what happens with $p'(u,v)$. Note, that after this permutation all variables $u$ and $v$ remains the same except $u^i$
which changes sign. Let us denote by $A$ the sum of all monomials of $p'$ which contain even number of variables from $u^{i}$ and
by $B$ the sum of all monomials of $p'$ which contain odd number of variables from $u^{i}$, so
$$
p'(u,v) = A(u,v) + B(u,v).
$$
Note that after our permutation of variables $A$ remains the same and $B$ changes sign.
Since $p'$ changes sign we have that the absolute value of $B$ is greater then the absolute value of $A$.
We proved that this happens for any input from $I$, so the sign of $f$ is determined by the sign of $B$ for all such inputs.
So we can erase all monomials from $A$ and the resulting polynomial will still sign-represent $f$ for all inputs from $I$.

Repeating this argument for all $i \in [d]$ we obtain a threshold gate $q$ for $f$ such that
each monomial of $q$ contains odd number of variables from each of the groups $u^{1} ,\ldots, u^{d}$.
But the degree of $q$ is at most $d$, so it is only possible that each monomial of $q$ contains one variable
from each of the groups $u^{1} ,\ldots, u^{d}$ (and no variables from $v^{1} ,\ldots, v^{d}$).

We have proved that the new polynomial $q$ sign-represents $f$ correctly for all inputs from $I$.
Now note, that for any input from $\zoon \setminus I$ there is an $i \in [d]$ such that $x^{i} = y^{i}$, or equivalently,
all variables from $u^{i}$ are zeros.
By the definition the value of the function $f$ on such input is $1$ and the value of $q$ is $0$, so
$q$ sign-represents $f$ for these inputs also.
\end{proof}

As a byproduct of the proof of this lemma we have the following corollary.
\begin{corollary}
$\signdeg(f) = d$.
\end{corollary}

Since $W(p') \geq W(q)$ it is enough to prove that
$$
W(q) \geq 2^{(k_d-2)\prod_{i=1}^{d-1}k_i}.
$$

Now we need a lemma concerning degree $1$ threshold gates for $\GT$.
The argument is quite standard (see~\cite{myhill-kautz61,parberry94mit}).

\begin{lemma} \label{lemma.gt}
Let
$
p = \sum_{i=1}^k w_i u_i
$
be a degree $1$ threshold gate for $\GT(x,y)$ where $x,y \in \zook$.
Then for $j \geq 2$
\begin{equation} \label{eq.gt.exp}
w_j \geq 2^{j-2} w_{1} > 0
\end{equation}
and
\begin{equation} \label{eq.gt.step}
w_{j} \geq w_{j-1}.
\end{equation}
\end{lemma}

\begin{proof}
We will actually prove that for each $j$ $w_j>0$ and for each $j\geq2$
\begin{equation} \label{eq.gt.proof}
w_{j} \geq \sum_{i=1}^{j-1} w_{i}.
\end{equation}

The inequality~\eqref{eq.gt.step} follows from this immediately
and the inequality~\eqref{eq.gt.exp} is easy to prove by induction.

To prove the inequality~\eqref{eq.gt.proof} we consider the specific input
$$
u = (\overbrace{-1, \ldots, -1}^{j-1}, 1, 0 \ldots, 0).
$$
It is easy to see that on this input the function $\GT$ is $1$, and this means
that $p(u)\geq 0$. It is easy to see that the inequality~\eqref{eq.gt.proof} follows.

To prove that $w_{j}>0$ just let $u_j=-1$ and $u_i = 0$ for $i \neq j$.
$\GT$ is zero for such input, so $p(u)<0$.
\end{proof}

It will be convenient for us to consider two variants of the $\GT$ function:
we denote by $GT_1$ the usual $\GT$ function and by $\GT_0$ the analogous function,
but now on the reversed input. That is, $\GT_0(x,y) =1$
if and only if $x \geq y$, where $x = (x_1, \ldots, x_k), y=(y_1,\ldots,y_k)$
are considered as binary representations of integer numbers where the most significant bits are
$x_{1}, y_{1}$.
It is easy to see that if
$$
p = \sum_{i=1}^k w_i u_i
$$
is a threshold gate for $\GT_0$ then we have
$$
w_{j-1} \geq w_{j}
$$
and
$$
w_{n-j+1} \geq 2^{j-2} w_{n} > 0,
$$
where $j = 2, \ldots n$.

Now we can prove the main lemma.

\begin{lemma} \label{lemma.main}
For all $l \leq d$ if $\alpha \in K$ is such that
$\num_{\alpha,i} \alpha_{i} = 1$ for all $i \geq l$ and
$\beta = (\alpha_1, \ldots, k_l - \alpha_l + 1, \ldots, \alpha_d)$.
Then
$$
w_{\beta} \geq w_{\alpha} 2^{(k_d-2)\prod_{i=l}^{d-1}k_i}.
$$
\end{lemma}
The idea is the following: we fix variables in all groups
$u^1, \ldots, u^d$ except one. Then the function $f$ becomes essentially the $\GT$ function
and we can apply Lemma~\ref{lemma.gt}. Repeating this trick we can accumulate the large factor
due to the specific construction of our order. More specifically the proof goes by the induction on the decreasing $l$.
For the base of the induction $l=d$ we fix all variables except $u^d$ and applying inequality~\eqref{eq.gt.exp}
immediately obtain the desired result. For the induction step we first apply the induction hypothesis to $l+1$
and then apply inequality~\eqref{eq.gt.step} to the $l$th coordinate. Then we can again apply induction hypothesis to $l+1$
and so forth. In this way we can apply induction hypothesis $k_l$ times and obtain the desired result.
We proceed to the detailed proof.

\begin{proof}
The proof goes by the induction on the decreasing $l$.

\paragraph{\textbf{The base of induction $l=d$.}}

We fix all variables $u^{i}$ except $u^{d}$ in the following way.
For any $i$ let $u^{i}_{\alpha_i} = 1$ and $u^{i}_{j} = 0$ for all $j \neq \alpha_{i}$.

Now we have a function in variables $u^{d}$ and it is not hard to see that this function coincides with either $\GT_0$, or $\GT_1$.
Applying the inequality~\eqref{eq.gt.exp} (or the corresponding inequality for $\GT_0$) for $j=k_d$ we obtain
$$
w_{\beta} \geq 2^{k_d - 2} w_{\alpha}.
$$

\paragraph{\textbf{Induction step.}}

To show the lemma for $l<d$
we repeat several times the following procedure consisting of two steps.
After the $i$th application of the procedure we will get a tuple $\alpha^i$ such that
$\num_{\alpha^i,l} \alpha^{i}_{l} = i+1$ and $\num_{\alpha^i,j} \alpha^{i}_{j} = 1$ for all $j>l$.
To unify the notation we denote $\alpha^0 = \alpha$.

During the procedure we will not change the values of the first $l-1$ coordinates.
This means that for all tuples we consider, the order corresponding to the $l$th coordinate is the same.
Let us assume without loss of generality that this order is $<_{1}$.
This in particular means that $\alpha_l = 1$.

We also will not change the coordinates $\alpha_{l+2}, \ldots, \alpha_{d}$.
Note, that this means that if $\num_{\alpha^{i},l+1} \alpha^{i}_{l+1}$ is odd (as in the beginning)
then all $$\num_{\alpha^i,l+2} \alpha^{i}_{l+2}, \ldots, \num_{\alpha^i,d} \alpha^{i}_{d}$$ are equal to $1$.

\textbf{Step 1.}
We apply the induction hypothesis for the coordinate $l+1$.
We have that for $\widetilde{\alpha}^{i+1} = (\alpha_1, \ldots, \alpha^{i}_l, k_{l+1} - \alpha^{i}_{l+1} + 1, \ldots, \alpha_d)$
$$
w_{\widetilde{\alpha}^{i+1}} \geq w_{\alpha^i} 2^{(k_d-2)\prod_{i=l+1}^{d-1}k_i}.
$$
Note that now the ordinal number of $l+1$st coordinate (w.r.t. the corresponding order) is $k_{l+1}$.

\textbf{Step 2.}
We fix all variables $u^{j}$ except $u^{l}$ in the following way:
for any $j$ let $u^{j}_{\widetilde{\alpha}^{i+1}_{j}} = 1$ and $u^{j}_{m} = 0$ for all $m \neq \widetilde{\alpha}^{i+1}_{j}$.
Now we have a function in the variables $u^{l}$ and it is not hard to see that this function coincides with $\GT_1$
(this happens because we agreed that the order corresponding to the $l$th component is $<_{1}$, if it were $<_{0}$ we would have $\GT_{0}$ here).
We apply the inequality~\eqref{eq.gt.step} to the coordinate $l$.
After that we get
$$
\alpha^{i+1} = (\alpha_1, \ldots, \alpha_l + i+1, k_{l+1} - \alpha^{i}_{l+1}+1, \ldots, \alpha_d)
$$
such that
$$
w_{\alpha^{i+1}} \geq w_{\widetilde{\alpha}^{i+1}}.
$$
Due to the rule~\eqref{eq.order} defining the order on the next component of the tuple we have that
the order on the $l+1$st component changes.
This means that the ordinal number of the $l+1$st coordinate w.r.t. the corresponding order is again $1$.
>From this we have, as we stated above, that for all $j \geq l+1$ it is true that $\num_{\alpha^{i+1},j} \alpha^{i+1}_{j} = 1$.
So now we are again in the position to apply Step $1$.

We repeat these two steps until the $l$th coordinate of $\widetilde{\alpha}^{i}$ reaches $k_l$
(in the end we repeat Step $1$, we are unable to repeat Step $2$ since the $l$th coordinate is already $k_l$ and can not be increased,
so in the end we get $\widetilde{\alpha}^{i}$ for suitable $i$).
Since the $l$th coordinate increases by $1$ at each iteration (on Step $2$) we can repeat Step $2$ $k_{l}-1$ times and Step $1$ $k_l$ times.
Thus in the end we get the vector $\widetilde{\alpha}^{k_{l}}$
and it is easy to see that
$$
w_{\widetilde{\alpha}^{k_{l}}} \geq w_{\alpha} \prod_{j=1}^{k_{l}} 2^{(k_d-2)\prod_{i=l+1}^{d-1}k_i} =
w_{\alpha} 2^{(k_d-2)\prod_{i=l}^{d-1}k_i}.
$$
Let us check that $\beta = \widetilde{\alpha}^{k_{l}}$.
It is easy to see that during this process only the coordinates $l$ and $l+1$ of $\alpha$ change.
In the end of the process coordinate $l$ has the number $k_l$ w.r.t. the order corresponding to the $l$th coordinate of $\alpha$.
This means that $\widetilde{\alpha}^{k_l}_{l} = k_l - \alpha_l + 1 = \beta_l$.
Note also that from this and from the evenness of $k_l$ we have that $\alpha_{l}$ and $\widetilde{\alpha}^{k_l}_{l}$
define different orders on the next coordinates.
Now let us see what happens with the coordinate $l+1$.
After the process the ordinal number of $\widetilde{\alpha}^{k_l}_{l+1}$
w.r.t. the order corresponding to the $(l+1)$st component of $\widetilde{\alpha}^{k_l}$ is equal to $k_{l+1}$.
This means that the ordinal number of $\widetilde{\alpha}^{k_l}_{l+1}$ w.r.t. the other ordering is $1$.
Since orders corresponding to the $(l+1)$st coordinates of $\alpha$ and $\widetilde{\alpha}^{k_l}$ are different
we have that $\widetilde{\alpha}^{k_l}_{l+1} = \alpha_{l+1} = \beta_{l+1}$ (recall that $\num_{\alpha,l+1} \alpha_{l+1} = 1$).
\end{proof}

It is easy to prove Theorem~\ref{theorem.main} now.
Applying Lemma~\ref{lemma.main} with $l=1$ we get
$$
w_{\beta} \geq w_{\alpha} 2^{(k_d-2)\prod_{i=1}^{d-1}k_i}.
$$
Now, it is easy to see that $w_{\alpha}>0$ (just substitute $u^{1}_{\alpha_1} = -1$, $u^{i}_{\alpha_i} = 1$
for all $i \neq 1$ and $u^{i}_{j} = 0$ for all $i$ and all $j \neq \alpha_i$).
We conclude that $w_{\alpha} \geq 1$ and
$$
w_{\beta} \geq 2^{(k_d-2)\prod_{i=1}^{d-1}k_i}.
$$

\section{Improved Lower Bound} \label{sec.better}

In this section we improve the argument of the previous sections to obtain the better lower bound.
More precisely we prove
$$
W \geq 2^{\Omega(2^{2n/5})}.
$$

We will work with Boolean variables $\moo$, so we change the definition of the $\sgn$-function:
$\sgn(x)=1$ if $x \geq 0$ and $\sgn(x)=-1$ otherwise.

The idea is to use another function instead of $\GT$ as a building block
in our construction.
Indeed,
we may try to carry out the proof
if we use any function $f \colon \moon \to \moo$
which can be defined in the following terms:
consider uniform linear forms $L_1(x), \ldots, L_{k}(x)$ and for any $x \in \moon$ let $f(x)$ be the sign of $L_i(x)$,
where $i$ is the maximal index such that $L_i(x)$ is nonzero.
Here uniformity of linear forms is required to make the symmetry argument of Lemma~\ref{lemma.symmetry} work.
Note also that we can use different functions on different coordinates
and in fact there is sense in it since the last coordinate plays a very different role than the other coordinates.
Actually, we have already used different functions in different coordinates when we choose numbers of variables $k_i$
to be different for different $i$.

Note that the function $\GT$ in coordinates $1, \ldots, d-1$ is not very economical.
Indeed, we use $2k$ variables and get only $k$ iterations in one coordinate
in the proof of Lemma~\ref{lemma.main} (that is, $k_l$ applications of Step~$1$ in the proof of this lemma).
So we can have about $n/2k$ coordinates with $k$ iterations each, so we have bound of approximately $2^{\Omega(k^{n/2k})}$.
If with some other function we can have more iterations with less variables, we will be able to improve the bound.

We prove that for any $k \geq 3$ there is a function which with the use of $k$ variables allows us to make $k-1$ iterations.
\begin{definition}
For $x = (x_1, \ldots, x_k) \in \moo^k$ let $g(x_{1}, \ldots, x_{k})$ be equal to $-x_k$ if the bits $x_{1}, \ldots, x_{k}$ are not all equal and let it be $x_k$ if they are all equal.
\end{definition}

This function can be easily represented by a linear threshold gate:
$$
g(x_1, \ldots, x_{k}) = \sgn(\sum_{i=1}^{k-1} x_i -(k-2)x_{k}).
$$

But to make $g$ to be suitable as a building block in our combinatorial construction we have to express it in terms of
linear forms $L_{1}(x), \ldots, L_{k}(x)$ (otherwise we are not able to define our ordering).
Consider $k$ linear form: $L_{i}(x) = x_{i} - x_{i+1}$ ($L_i$ will play a role of $u^i$ in the previous proof) for $i = 1, \ldots, k-1$ and $L_{0}(x) = x_{1} + x_{k}$.
The alternative (equivalent) definition of $g$ is that $g(x)$ is equal to the sign of the last nonzero in the sequence
$L_{0}(x), L_{1}(x), L_{2}(x), \ldots, L_{k-1}(x)$.
Note that now it is more convenient for us to start the numeration of $L_{i}$ from $0$.
The reason for this is that the actual benefit we will get only from linear forms $L_{1},\ldots, L_{k-1}$.
The form $L_{0}$ is needed only for technical reasons (the same role was previously played by the variables $v^{i}$).

Note that in the proof of the weaker bound we needed actually not one base function, but two of them.
They were very similar though: $\GT_0$ and $\GT_1$.
In the case of the stronger bound two functions will differ more substantially.
Again, $g_1$ is just the function $g$ we defined above.
As for $g_0$ we let $g_{0}(x)$ be $x_{1}$ if not all bits of the input are equal and $g_{0}(x) = - x_{1}$
if all bits of input are equal. That is now we not only reverse the order of variables but also
multiply the value of the function by $-1$.
Note that $g_{0}(x)$ is equal to the sign of the last nonzero in the sequence
$- L_{0}(x), L_{k-1}(x), L_{k-2}(x), \ldots, L_{1}(x)$.

Now we can apply the previous proof scheme with the functions $g_1$ and $g_0$ on the first $d-1$ components and with $\GT$
on the last component. We denote the new function by $f$ again. Below we state what changes in the proof.

For the new function the construction of the ordering is the same except that we use different orderings $<^{'}_0$ and $<^{'}_1$
on the first $d-1$ coordinates,
namely we let $<^{\prime}_1$ to be $0, 1, 2, \ldots, k-2, k-1$ and $<^{\prime}_{0}$ to be $0, k-1, k-2, \ldots, 2, 1$,
that is, $0$ is always the smallest element. The orderings on the last coordinate remains the same as before (as well as the rule~\eqref{eq.order}  defining the ordering on each next coordinate).

In the Definition~\ref{def.func} we now have only variables $x^{1}, x^{2},\ldots, x^{d}, y^{d}$
and we let
$$
f(x)=\sgn( (-1)^{c_1 + \ldots + c_{d-1}} L_{\alpha_1}(x^1)L_{\alpha_{2}}(x^2) \ldots L_{\alpha_{d-1}}(x^{d-1})(x_{\alpha_{d}}^{d}-y_{\alpha_{d}}^{d}))
$$
for the largest $\alpha$ for which the expression is nonzero,
where $c_i = 1$ if $\alpha_{i} = 0$ and the order corresponding to the $i$-th coordinate of $\alpha$ is $<^{\prime}_{0}$
and $c_i = 0$ otherwise.
If there is no such $\alpha$ (which can happen only if $x^d=y^d$) we let $f(x)$ to be $1$.
Note that the number of variables  $n = \sum_{i=1}^{d-1}k_i + 2k_d$ is almost twice less~than~before.

The theorem we prove has the following form.
\begin{thm} \label{theorem.main.strong}
Let $k_i \geq 3$ be odd for $i < d$ and $k_d \geq 3$. Then
$$
\signdeg(f) = d
$$
$$
W(f,d)\geq 2^{(k_d-2)\prod_{i=1}^{d-1}(k_i-1) - d \log n}
$$
\end{thm}

The proof of the theorem follows the same lines with minor changes.
For the sake of completeness we present the details.



Let us consider an arbitrary threshold gate $p$
for $f$ of degree at most $d$.
That is, for any $x^1, \ldots, x^{d-1}, x^d , y^d$ we have
$$
f(x^1, \ldots, x^{d-1}, x^d , y^d) = \sgn (p(x^1, \ldots, x^{d-1}, x^d , y^d)).
$$
It will be convenient for us to work in the variables
\begin{equation*}
u^i_{j} = L_{j}(x^{i})
\end{equation*}
for $i = 1, \ldots, d-1$ and
\begin{equation*}
u^{d}_{j} = x^{d}_j - y^{d}_{j},\ \ \ \ \ \ \ v^{d}_{j} = x^{d}_j + y^{d}_{j}.
\end{equation*}
So after substituting
$$
x^{i}_{j} = (u^i_0 -u^i_1 - \ldots - u^i_{j-1} + u^{i}_{j} + \ldots + u^i_{k-1})/2
$$
for $i= 1, \ldots, d-1$
and
$$
x^{d}_j = ({u^{d}_j + v^{d}_j})/{2},\ \ \ \ y^{d}_j = ({u^{d}_j - v^{d}_j})/{2}
$$
and multiplying the polynomial by $2^d$ to make the coefficients integer we obtain
the polynomial $p'$ in the variables $u^{i}_j,\ v^{d}_j$ that sign-represents $f$. That is
$$
f(x^1, \ldots, x^{d-1}, x^d , y^d) = \sgn(p'(u^1, \ldots, u^{d-1}, u^d , v^d)).
$$
It is easy to see that the weight of the new polynomial is almost the same as the weight of $p$
(compared to the value of our bound). Namely, we have the following bound.

\begin{lemma} 
$W(p') \leq n^d W(p)$.
\end{lemma}

\begin{proof}
Consider one monomial $g$ of $p$ of degree $l \leq d$.
We can think of the transformation above as of substituting each variable by the sum
of at most $\max_i k_i \leq n$ variables and then multiplying everything by $2^{d-l}$.
After opening the brackets we will have at most $n^{l}$ new monomials each of weight $2^{d-l}$,
so the overall weight of the new monomials coming from the monomial $g$ is at most $n^{d}$ times the coefficient of $g$.
Since this happens for all monomials, the upper bound on the weight of the polynomial $p'$ as stated follows.
\end{proof}

Next we prove that we can assume that $p'$ has a nice structure.
\begin{lemma} 
If we substitute by $0$ all coefficients of the monomials of $p'$ in which the variables from one of the groups $u^{1}, \ldots, u^{d}$ do not appear,
the resulting polynomial $q$ will also sign-represent $f$.
\end{lemma}

\begin{proof}
Let $I$ be the set of inputs $(x,y)$ to $f$ such that $x^d \neq y^d$.
This condition implies that for each such input there is an
$\alpha = (\alpha_1, \ldots, \alpha_d) \in K = [k_1-1]_0 \times \ldots  \times [k_{d-1}-1]_0 \times [k_d]$ such that
$$
u^1_{\alpha_1}u^2_{\alpha_{2}} \ldots u^{d-1}_{\alpha_{d-1}}(x_{\alpha_{d}}^{d}-y_{\alpha_{d}}^{d}) \neq 0,
$$
where by $[k]_0$ we denote the set $\{0, \ldots k\}$.
This means that the value of the function is determined by the sign of this product for the largest such $\alpha$.

Consider an arbitrary input $(x,y)$ from $I$. Let us first permute the variables $x^{d}$ with the variables $y^{d}$.
It is not hard to see that after such permutation the function $f$ changes sign, and so $p'$ should also change sign.

Now let us see what happens with $p'(u^1, \ldots, u^d ,v^d)$. Note, that after this permutation all variables $u$ and $v^d$ remains the same except $u^d$
which changes sign. Let us denote by $A$ the sum of all monomials of $p'$ which contain even number of variables from $u^{d}$ and
by $B$ the sum of all monomials of $p'$ which contain odd number of variables from $u^{d}$, so
$$
p'(u,v^d) = A(u,v^d) + B(u,v^d).
$$
Note that after our permutation of variables $A$ remains the same and $B$ changes sign.
Since $p'$ changes sign we have that the absolute value of $B$ is greater then the absolute value of $A$.
We proved that this happens for any input from $I$, so the sign of $f$ is determined by the sign of $B$ for all such inputs.
So we can erase all monomials from $A$ and the resulting polynomial will still sign-represent $f$ for all inputs from $I$.

Now consider an arbitrary $i \in [d-1]$. Let us substitute the vector $x^{i}$ by the vector $-x^{i}$.
Again, it is not hard to see that after such substitution the function $f$ changes sign, and so $p'$ also should change sign.

Let us see what happens with $p'(u^1, \ldots, u^d ,v^d)$. Note, that after this substitution all variables $u$ and $v^d$ remains the same except $u^i$
which changes sign. Let us denote by $A$ the sum of all monomials of $p'$ which contain even number of variables from $u^{i}$ and
by $B$ the sum of all monomials of $p'$ which contain odd number of variables from $u^{i}$, so
$$
p'(u,v^d) = A(u,v^d) + B(u,v^d).
$$
After our substitution $A$ remains the same and $B$ changes sign.
Again, we can erase all monomials from $A$ and the resulting polynomial will still sign-represent $f$ for all inputs from $I$.

Repeating this argument for all $i \in [d-1]$ we obtain a threshold gate $q$ for $f$ such that
each monomial of $q$ contains odd number of variables from each of the groups $u^{1} ,\ldots, u^{d}$.
But the degree of $q$ is at most $d$, so it is only possible that each monomial of $q$ contains one variable
from each of the groups $u^{1} ,\ldots, u^{d}$ (and no variables from $v^{d}$).

We have proved that the new polynomial $q$ sign-represents $f$ correctly for all inputs from $I$.
Now note, that for any input from $\{-1,1\}^n \setminus I$ we have $x^{d} = y^{d}$, or equivalently,
all variables from $u^{d}$ are zeros.
By the definition the value of the function $f$ on such input is $1$ and the value of $q$ is $0$, so
$q$ sign-represents $f$ for these inputs also.
\end{proof}

As a byproduct of the proof of this lemma we have the following corollary.
\begin{corollary}
$\signdeg(f) = d$.
\end{corollary}

Since $W(p') \geq W(q)$ it is enough to prove that
$$
W(q) \geq 2^{(k_d-2)\prod_{i=1}^{d-1}(k_i-1)}.
$$

We present now the analog of Lemma~\ref{lemma.gt}.

\begin{lemma} \label{lemma.gt.strong}
Let
$
p = \sum_{i=0}^{k-1} w_i u_i
$
be a degree $1$ threshold gate for $g_1(x)$ where $x \in \mook$.
Then for $j = 0, 1, \ldots, k-1$ we have
$
w_j > 0
$
and for $j=2,\ldots, k-1$ we have
$
w_{j} > w_{j-1}
$.
\end{lemma}

For the function $g_0$ analogous statement is true.

\begin{lemma} \label{lemma.gt.strong2}
Let
$
p = \sum_{i=0}^{k-1} w_i u_i
$
be a degree $1$ threshold gate for $g_0(x)$ where $x \in \mook$.
Then we have $w_0 < 0$, for $j = 1, \ldots, k-1$ we have
$
w_j > 0$
and for $j=2,\ldots, k-1$ we have
$
w_{j-1} > w_{j}.
$
\end{lemma}

\begin{proof}[Proof of Lemma~\ref{lemma.gt.strong}]
To show the first part of the lemma for $j \neq 0$ let
$x_1 = \ldots = x_{j} = -1$ and $x_{j+1} = \ldots = x_{k} = 1$.
Then $g_{1}(x) = -1$ and thus $p(x) = -2w_{j} < 0$.

For $j=0$ let $x_{1} = \ldots = x_{k} = -1$.
Then again $g_{1}(x) = -1$ and $p(x) = -2w_{0} <0$.

To show the second part let $x_{j} = -1$ and $x_{l}=1$ for $l \neq j$.
Then we have $g_{1}(x)=-1$ and $p(x) = 2(w_{j-1} - w_{j} + w_{0}) <0$.
And thus $w_{j} > w_{j-1} + w_{0} > w_{j-1}$.
\end{proof}

\begin{proof}[Proof of Lemma~\ref{lemma.gt.strong2}]
For $j \neq 0$ letting $x_{1}=\ldots=x_{j}=-1$, $x_{j+1}=\ldots=x_{k}=1$ we have $- w_{j} < 0$.
Letting $x_{1}=\ldots=x_{k}=1$ we have $w_{0} < 0$ (note that we have $-L_{0}$ in the sequence defining $g_0$).
For $j=2,\ldots,k-1$ letting $x_{j}=1$ and $x_{l}=-1$ for $l \neq j$ we have $-w_{j-1} + w_{j} - w_{0} < 0$
and thus $w_{j-1} > w_{j}$ since $w_{0}$ is negative.
\end{proof}

The analog of Lemma~\ref{lemma.main} is very similar to the previous version,
but becomes a little bit clumsy since we distinguish cases of $l=d$ and $l<d$.

\begin{lemma} \label{lemma.main.strong}
For all $l \leq d$ if $\alpha \in K$ is such that
$
\num_{\alpha,i} \alpha_{i} = 1
$
for all $i \geq l$ and
$\beta = (\alpha_1, \ldots, k_l - \alpha_l + \delta_{l,d}, \ldots, \alpha_d)$,
where $\delta_{l,d}$ is a Kronecker delta (that is $\delta_{ij}=1$ if $i=j$ and $\delta_{ij} = 0$ otherwise)
then
$$
w_{\beta} \geq w_{\alpha} 2^{(k_d-2)\prod_{i=l}^{d-1}(k_i-1)}.
$$
\end{lemma}

Concerning the proof of the lemma, the base of the induction remains completely the same (note, that the statement is the same also).
As for the induction step, it also remains the same but now we can apply the induction hypothesis $k_l-1$ times instead of $k_l$ times in the previous proof.
For the sake of completeness we present the proof.
\begin{proof}
This proof repeats the proof of Lemma~\ref{lemma.main} almost literally.

\noindent
The proof goes by induction on decreasing $l$.

\paragraph{\textbf{The base of induction $l=d$.}}

We fix all variables $u^{i}$ except $u^{d}$ in the following way.
For any $i$ let $u^{i}_{\alpha_i} = 1$ and $u^{i}_{j} = 0$ for all $j \neq \alpha_{i}$
(see the proof of Lemma~\ref{lemma.gt.strong} on how to do this).

Now we have a function in variables $u^{d}$ and it is not hard to see that this function coincides with either $\GT_0$, or $\GT_1$.
Applying the inequality~\eqref{eq.gt.exp} (or the corresponding inequality for $\GT_0$) for $j=k_d$ we obtain
$$
w_{\beta} \geq 2^{k_d - 2} w_{\alpha}.
$$

\paragraph{\textbf{Induction step.}}

To show the lemma for $l<d$
we repeat several times the following procedure consisting of two steps.
After the $i$th application of the procedure we will get a tuple $\alpha^i$ such that
$\num_{\alpha^i,l} \alpha^{i}_{l} = i+1$ and $\num_{\alpha^i,j} \alpha^{i}_{j} = 1$ for all $j>l$.
To unify the notation we denote $\alpha^0 = \alpha$.

During the procedure we will not change the values of the first $l-1$ coordinates.
This means that for all tuples we consider, the order corresponding to the $l$th coordinate is the same.
Let us assume without loss of generality that this order is $<^{\prime}_{1}$.
This in particular means that $\alpha_l = 1$.

We also will not change the coordinates $\alpha_{l+2}, \ldots, \alpha_{d}$.
Note, that this means that if $\num_{\alpha^{i},l+1} \alpha^{i}_{l+1}$ is odd (as in the beginning)
then all
$$
\num_{\alpha^i,l+2} \alpha^{i}_{l+2}, \ldots, \num_{\alpha^i,d} \alpha^{i}_{d}
$$
are equal to $1$.

\textbf{Step 1.}
We apply the induction hypothesis for the coordinate $l+1$.
We have that for $\widetilde{\alpha}^{i+1} = (\alpha_1, \ldots, \alpha^{i}_l, k_{l+1} - \alpha^{i}_{l+1} + \delta_{l+1,d}, \ldots, \alpha_d)$
$$
w_{\widetilde{\alpha}^{i+1}} \geq w_{\alpha^i} 2^{(k_d-2)\prod_{i=l+1}^{d-1}(k_i-1)}.
$$
Note that now the ordinal number of $l+1$st coordinate (w.r.t. the corresponding order) is $k_{l+1}-1 + \delta_{l+1,d}$.

\textbf{Step 2.}
We fix all variables $u^{j}$ except $u^{l}$ in the following way:
for any $j$ let $u^{j}_{\widetilde{\alpha}^{i+1}_{j}} = 1$ and $u^{j}_{m} = 0$ for all $m \neq \widetilde{\alpha}^{i+1}_{j}$.
Now we have a function in the variables $u^{l}$ and it is not hard to see that this function coincides with $g_1$
(this happens because we agreed that the order corresponding to the $l$th component is $<^{\prime}_{1}$, if it were $<^{\prime}_{0}$ we would have $g_{0}$ here).
We apply the inequality from Lemma~\ref{lemma.gt.strong} to the coordinate $l$.
After that we get $\alpha^{i+1} = (\alpha_1, \ldots, \alpha_l + i+1, k_{l+1} - \alpha^{i}_{l+1}+\delta_{l+1,d}, \ldots, \alpha_d)$
such that
$$
w_{\alpha^{i+1}} \geq w_{\widetilde{\alpha}^{i+1}}.
$$
Due to the rule~\eqref{eq.order} defining the order on the next component of the tuple we have that
the order on the $l+1$st component changes.
This means that the ordinal number of the $l+1$st coordinate w.r.t. the corresponding order is again $1$.
>From this we have, as we stated above, that for all $j \geq l+1$ it is true that $\num_{\alpha^{i+1},j} \alpha^{i+1}_{j} = 1$.
So now we are again in the position to apply Step $1$.

We repeat these two steps until the $l$th coordinate of $\widetilde{\alpha}^{i}$ reaches $k_l -1$
(in the end we repeat Step $1$, we are unable to repeat Step $2$ since the $l$th coordinate is already $k_l - 1$ and can not be increased,
so in the end we get $\widetilde{\alpha}^{i}$ for suitable $i$).
Since the $l$th coordinate increases by $1$ at each iteration (on Step $2$) we can repeat Step $2$ $k_{l}-2$ times and Step $1$ $k_l-1$ times.
Thus in the end we get the vector $\widetilde{\alpha}^{k_{l}-1}$
and it is easy to see that
$$
w_{\widetilde{\alpha}^{k_{l}-1}} \geq w_{\alpha} \prod_{j=1}^{k_{l}-1} 2^{(k_d-2)\prod_{i=l+1}^{d-1}(k_i-1)} =
w_{\alpha} 2^{(k_d-2)\prod_{i=l}^{d-1}(k_i-1)}.
$$

Let us check that $\beta = \widetilde{\alpha}^{k_{l}-1}$.
It is easy to see that during this process only the coordinates $l$ and $l+1$ of $\alpha$ changes.
In the end of the process coordinate $l$ has the number $k_l-1$ w.r.t. the order corresponding to the $l$th coordinate of $\alpha$.
This means that $\widetilde{\alpha}^{k_l-1}_{l} = k_l - \alpha_l = \beta_l$.
Note also that from this and from the oddness of $k_l$ we have that $\alpha_{l}$ and $\widetilde{\alpha}^{k_l-1}_{l}$
define different orders on the next coordinates.
Now let us see what happens with the coordinate $l+1$.
After the process the ordinal number of $\widetilde{\alpha}^{k_l-1}_{l+1}$
w.r.t. the order corresponding to the $(l+1)$st component of $\widetilde{\alpha}^{k_l-1}$ is equal to $k_{l+1} - 1 + \delta_{l+1,d}$.
This means that the ordinal number of $\widetilde{\alpha}^{k_l-1}_{l+1}$ w.r.t. the other ordering is $1$.
Since orders corresponding to the $(l+1)$st coordinates of $\alpha$ and $\widetilde{\alpha}^{k_l-1}$ are different
we have that $\widetilde{\alpha}^{k_l-1}_{l+1} = \alpha_{l+1} = \beta_{l+1}$ (recall that $\num_{\alpha,l+1} \alpha_{l+1} = 1$).
\end{proof}

To conclude the proof of our lower bound we have to choose the values of $k_{i}$ to maximize the lower bound we have.
It is not hard to see that the optimal way is to take $k_d = 3$ as before and to take $k_1 = k_{2} = \ldots = k_{d-1}$,
let us denote the value of them by $k$. Then the exponent of our bound will be about $(k-1)^{n/k}$.
Simple analysis shows that the maximum (over integers) is attained when $k=5$.
Thus we have a lower bound of $2^{2^{2(n-6)/5} - n}$.

To prove the first part of Theorem~\ref{theorem.high.deg} we can just let $m = \frac{5}{4} \eps n$ and
consider the function from the previous paragraph with $m$ variables. Then we have the lower bound $2^{\Omega(2^{2m/5})}$
for the degree $m/5$ threshold gates and applying the observation preceding Theorem~\ref{theorem.high.deg} we get the desired bound.
For the second part of the theorem let $m = \frac{5}{2} (1 + \eps) \log n$.

Finally we note that the result of~\cite{ppi09weights} can also be reproved by the same argument and with better constants if we
use H{\aa}stad's function in the last coordinate and $g$ function in other coordinates.

{\small
\bibliographystyle{abbrv}
\bibliography{bib/perceptron,bib/podolskii,bib/threshold,bib/sherstov,bib/general,bib/fourier}
}

\end{document}